\documentclass[12pt,reqno]{amsart}
\usepackage{amssymb,amsmath}
\usepackage[english]{babel}
\usepackage{amsfonts}
\usepackage{color}
\usepackage{amsthm}
\usepackage[colorlinks=true,linkcolor=NavyBlue, citecolor=NavyBlue,urlcolor=NavyBlue]{hyperref}
\usepackage{graphicx}
\usepackage{color}
\usepackage{mathtools}
\usepackage[usenames,dvipsnames]{pstricks}
\usepackage{bm}
\usepackage{amsmath}
\usepackage{url}
\usepackage{setspace}

\usepackage{tikz}
\usetikzlibrary{positioning,shapes,shadows,arrows,fadings}
\pgfdeclarelayer{nodelayer}
\pgfdeclarelayer{edgelayer}
\pgfsetlayers{nodelayer,edgelayer}
\usetikzlibrary{positioning,shapes,shadows,arrows,fadings}
\tikzstyle{user}=[circle,draw=gray!60,fill=gray!30,thick,scale=1.5, inner sep=0pt, minimum size=5mm]
\tikzstyle{userselected}=[circle,draw=red!130,fill=gray!30,thick,scale=1.5, inner sep=0pt, minimum size=5mm]
\tikzstyle{view}=[circle,draw=orange!100,fill=yellow!75,thick,scale=1.5, inner sep=0pt, minimum size=5mm]
\tikzstyle{view2}=[circle,draw=black!100,fill=red!50!,thick,scale=1.5, inner sep=0pt, minimum size=5mm]
\tikzstyle{arco}=[<->,shorten >=1pt,>=stealth',semithick]
\tikzstyle{semiarco}=[->,shorten >=1pt,>=stealth',semithick]

\DeclareMathOperator{\GL}{GL}

\DeclareMathOperator{\F}{\mathbb{F}}

\DeclareMathOperator{\ANT}{ANT}

\DeclareMathOperator{\CA}{CA}
\DeclareMathOperator{\LCD}{LCD}
\DeclareMathOperator{\Pp}{P}
\DeclareMathOperator{\Ss}{S}
\DeclareMathOperator{\ET}{ECTAK}
\DeclareMathOperator{\ETS}{ECTAKS}
\DeclareMathOperator{\EC}{\mathcal C}
\DeclareMathOperator{\Hh}{H}
\DeclareMathOperator{\KDF}{KDF}
\DeclareMathOperator{\Enc}{Enc}
\DeclareMathOperator{\Dec}{Dec}
\DeclareMathOperator{\Aa}{Adv}
\DeclareMathOperator{\rk}{rk}

\newcommand{\ANTi}[1]{\ANT_{#1}}
\newcommand\deq{\mathrel{\stackrel{\makebox[0pt]{\mbox{\normalfont\tiny def}}}{=}}}
\newcommand\ra{\mathrel{\stackrel{\makebox[0pt]{\mbox{\normalfont\tiny R}}}{\in}}}

\newcommand{\bk}{{\bf k}}
\newcommand{\bt}{{\bf t}}
\newcommand{\bmm}{{\bf m}}
\let\phi\varphi

\theoremstyle{plain}
\newtheorem{theorem}{Theorem}[section]

\newtheorem{lemma}[theorem]{Lemma}

\theoremstyle{remark}
\newtheorem{remark}[theorem]{Remark}
\theoremstyle{problem}
\newtheorem{problem}[]{Problem}
\theoremstyle{definition}
\newtheorem{example}[theorem]{Example}
\newtheorem*{notation*}{Notation}

\begin{document}
\title[An Authenticated Key Scheme over EC]{An Authenticated Key Scheme over Elliptic Curves for Topological Networks}

\author{R. Aragona} 
\author{R. Civino} 
\author{N. Gavioli} 
\author{M. Pugliese} 
\address{DISIM - University of L'Aquila - Italy}
\email{riccardo.aragona@univaq.it} 
\email{roberto.civino@univaq.it} 
\email{norberto.gavioli@univaq.it} 
\email{marpug@univaq.it} 
\thanks{R. Aragona, R. Civino, and N. Gavioli are members of INdAM-GNSAGA (Italy). N. Gavioli is member of the Centre of excellence ExEMERGE at University of L'Aquila which also partially funds R. Civino.}

\date{}

\begin{abstract}
Nodes of sensor networks may be resource-constrained devices, often having a limited lifetime, making sensor networks remarkably dynamic environments. Managing a cryptographic protocol on such setups may require a disproportionate effort when it comes to update the secret parameters of new nodes that enter the network in place of dismantled sensors. For this reason, the designers of schemes for sensor network are always concerned with the need of scalable and adaptable solutions. In this work, we present a novel elliptic-curve based solution, derived from the previously released cryptographic protocol TAKS, which addresses this issue. We give a formal description of the scheme, built on a two-dimensional vector space over a prime field and over elliptic curves, where node topology is more relevant  than node identity, allowing a dynamic handling of the network and reducing the cost of network updates. We compare our solution with classical Diffie-Hellman-like protocols and we also study some security concerns and their relation to the related discrete logarithm problem over elliptic curves.

\end{abstract}

\maketitle


\normalsize
\section{Introduction}
The \emph{Topology Authenticated Key Scheme} (TAKS) is a cryptographic protocol, proposed in~\cite{phdpugliese,pugliese08} for the first time and successively generalised~\cite{pugl13,pugliese13}, providing security over a resource-constrained network (typically ad-hoc networks, e.g.~sensor networks for monitoring services).
Its authentication mechanism is based on \emph{node topology} rather than on \emph{node identity}, due to the limited lifetime of nodes in a resource-constrained network. Indeed, while nodes in infra-structured networks can rely on a external power supply and on a stable planned maintenance service with human intervention, the nodes in ad-hoc networks can rely only on their own battery or some other energy harvesting mechanism, and maintenance services are usually remotely performed without human intervention. When an off-duty node is replaced with a new node, the new node identity enters in the network, but node topology remains unchanged and the authentication mechanism does not need any updating. Other examples of network-based key pre-distribution schemes may be found in~\cite{alsubaie,arjmandi,haowen,Kuchipudi,liuning}. \\


 The scheme we propose in this paper, called \emph{Elliptic Curve based Topological Authenticated Key Scheme} ($\ETS$), is derived from TAKS and is defined as a hybrid deterministic \emph{Key Establishment Protocol} (KEP) over elliptic curves, and is designed to establish both point-to-point and point-to-multipoint secure links among nodes. Security features of $\ETS$ may include confidentiality (data encryption), data integrity and sender authentication (signature). Other examples of hybrid KEPs may be found in \cite{kavitha,kimyun,manjunath}. 
 In $\ETS$ the shared secret is a symmetric key generated by each party involved in the communication session upon a successful authentication process, where each party verifies if the other party belongs to its authenticated network. 
 Such a network is represented as a graph, where parties (network nodes) are the vertices and the communication links are the edges. The assignment parameter to each node is carried out by an external \emph{Certification Authority}  ($\CA$); the scheme parameter are successively preloaded into the nodes. 
 
\noindent While TAKS~\cite{pugliese08} provides only key establishment facilities for point-to-point communications by means of a Diffie-Hellman-like scheme and its generalisation~\cite{pugl13} extends to point-to-multipoint sessions, the improvements  implemented in this paper directly provide $\ETS$ with key establishment capabilities for both point-to-point and point-to-multipoint communications. More importantly, elliptic curve cryptography allows achieving comparable security levels using reduced key lenghts~\cite{gupta2002performance}.\\

 In this paper we provide a rigorous description of $\ETS$ and address a security analysis of the scheme. In this regard, we will show that $\ETS$ can be broken if an adversary can solve the intractable Discrete Logarithm problem over elliptic curves, provided that it also manage to solve a linear system of equations.\\
 The paper is organised as follows: in Sec.~\ref{sec:nota} we introduce the notation and some auxiliary results. The scheme is defined in Sec.~\ref{sec:scheme}, together with the authenticated-encryption methods which is derived from it. An early security analysis of the scheme is carried out in Sec.~\ref{sec:security}. Sec.~\ref{sec:concl} concludes the paper with some considerations on open problems.

\section{Notation}\label{sec:nota}
$\ETS$ is a scheme based on elliptic curves as well as vector spaces over finite fields. The network of nodes where $\ETS$ is built is modeled by a graph. In order to provide a rigorous description of the scheme and of its model, let us define the following elements.
\subsubsection*{Spaces}
Let $p$ be a prime number, and let $\F_p$ be the finite field with $p$ elements. We denote by $\alpha \ra \F_p$ an uniformly random generated  element in $\F_p$. The scheme presented is this paper is mainly based on the 2-dimensional vector space $(\F_p)^2$. Scalar elements in $\F_p$ are usually denoted by lower-case greek letters, whereas vectors in $(\F_p)^2$ are denoted by bold latin letters. Given two vectors $\bk = (\alpha_1, \alpha_2), \bt = (\beta_1, \beta_2) \in (\F_p)^2$ we define the scalar product over $\F_p $ of $\bk$ and $\bt$ as
\[
\bk \cdot \bt \deq \alpha_1\beta_1+\alpha_2\beta_2 \in \F_p.
\]
\subsubsection*{Elliptic curves}
Let $q > 3$ be the power of a prime number. An elliptic curve over $\F_q$ is the abelian group of $\F_q$-rational points satisfying a Weierstrass equation~\cite{Sil09}, i.e. 
\[
\EC \,\deq \{(x,y) \in (\F_q)^2 \mid y^2 = x^3+ax+b\}, 
\]
where $a, b \in \F_q$. 
Throughout this paper we will denote by $G\in \EC$ the generator of a subgroup of $\EC$ of order $p$, called the \emph{base point}. Given a vector $\bk = (\alpha_1, \alpha_2) \in (\F_p)^2$ and $\beta_1, \beta_2 \in \F_p$ we define 
\[
\bk G \,\deq (\underbrace{G+G+\ldots +G}_{\alpha_1 \text{ times }},\underbrace{G+G+\ldots +G}_{\alpha_2 \text{ times }}) = (\alpha_1 G , \alpha_2 G) \in \EC^2, 
\]
and 
\[
\bk \cdot (\beta_1G, \beta_2G) \deq \alpha_1\beta_1G+\alpha_2\beta_2G \in \EC. 
\]
The following result is easily checked.
\begin{lemma}
Let $\bk,\bt \in (\F_p)^2$. Then we have 
\[
(\bk \cdot \bt ) G = \bk \cdot  (\bt  G).
\]
\end{lemma}
\begin{proof}
Let $\bk = (\alpha_1, \alpha_2)$  and $\bt = (\beta_1, \beta_2)$. Then 
\begin{align*}
(\bk \cdot \bt ) G &= (\alpha_1\beta_1+\alpha_2\beta_2)G  \\
&= \alpha_1\beta_1G+\alpha_2\beta_2G  \\
&=  (\alpha_1,\alpha_2) \cdot (\beta_1 G, \beta_2G)  \\
&= \bk\cdot (\bt G).  \qedhere
\end{align*}
\end{proof}
We assume the security of the scheme we propose in this paper to be relying on the following security assumption~\cite{KoblitzECC,MillerECC}.
\begin{problem}[ECDL problem]
Let $\EC$ be an elliptic curve over the finite field $\F_q$, let $P \in \EC$ and $m > 1$. The \emph{(computational) Elliptic Curve Discrete Logarithm problem} (ECDL problem)
is the problem of finding the integer $m$ when the points $P$ and $mP$ are given.
\end{problem} 
Provided that the curve meets some well-established requirements (see e.g~\, \cite{Bernsteincurve,FIPS186,Koblitzcurve,weakfields}), the ECDL problem is assumed to be computationally intractable~\cite{Bach,koblitz2010intractable}. For some overviews on algorithms solving the ECDL problem and related problems, see e.g.~\cite{Galbrath,menstate}.
\subsubsection*{Graphs}
A \emph{directed graph} is a pair of sets $(V,E)$, where $V \neq \emptyset$ and $E \subseteq V\times V$. We call \emph{arrows} the elements of $E$, and for each arrow $e = (i,j) \in E$ we denote the \emph{tail} of the arrow as $t(e)\deq i$ and the \emph{head} of the arrow as $h(e)\deq j$. In order to keep the notation cleaner, we will sometimes denote the arrow $(i,j)$ by writing ``$i-j$". If $i,j \in V$, we say that $i$ and $j$ are \emph{connected} in $(V,E)$ if $(i,j) \in E$ and $(j,i) \in E$. 
\section{The scheme}\label{sec:scheme}
The scheme $\ETS$ is a cryptographic protocol based on a network of users who want to communicate with each others.  The network is modeled by means of a graph where users are represented as nodes. A node models whatever physical device equipped with a processing unit together to sensors/detectors and a radio chip for {TX/RX operations}. Two users can communicate if and only if they are connected in the corresponding graph. In this setting, if the arrow from node $i$ to node $j$ exists, then user $i$ is allowed to start a communication session  with user $j$.
The communication between two users can start if they manage to exchange a shared secret which they will use to create an authenticated-encryption session (see Sec.~\ref{sec:auth}).  The external $\CA$ assigns a set of parameters, called \emph{Local Configuration Data} ($\LCD$), to each node of the network. For each node, $\LCD$ is composed by two secret components, that remain unchanged once generated, and a public component, which is updated every time a new node joins the dynamic network.
Moreover,  the $\CA$ generates $\LCD$ in such a way communications among nodes are allowed only if their topology is compliant to the planned network topology.\\

The scheme $\ETS$ presents some relevant features which makes it quite different from classical Diffie-Hellman-like schemes based on elliptic curves (ECDH schemes).
\begin{enumerate}
\item With respect to classic ECDH-like solutions, the scheme proposed in this paper does not only rely on pre-distribution of keys in nodes. As a matter of fact, the proposed protocol is rather based on a dynamic assignment of the public components to each node and on a static assignment of secret components. The shared secret in $\ETS$ is in fact a function of both sender and receiver private key components, while in ephemeral ECDH-like schemes the shared secret is usually a function of the complete sender private key and receiver public key and this is a nice property from a robustness point of view.
\item Vector spaces, rather than scalars, are introduced to allow the setup of authenticated truly point-to-multipoint communication sessions, i.e.~sessions from a master root node to leaf nodes,  without setting up multiple point-to-point sessions. The added dimensionality introduces a further degree of freedom for the CA in the procedure of parameters computation and assignment. From an engineering point of view, point-to-multipoint sessions provide a relevant feature especially for clustered networks, when management services, such as the updating of some configuration parameters in a specific cluster of the network, are activated.
\item The activation of an authenticated point-to-multipoint communication session is truly scalable in $\ETS$, since any new communication link from the master node to a new member  added to an already established point-to-multipoint session will not imply any parameter update. This relevant feature is due to the fact that the generated shared secret for any new communication link results the same shared secret of the existing point-to-multipoint session, differently from classical ECDH-like schemes, where the addition of a communication link implies a newly generated and different shared secret, which depends on the public key of the new member node. Therefore, a point-to-multipoint communication session results in the aggregation of multiple point-to-point sessions, each one from the master to a member. From an engineering point of view, scalable point-to-multipoint sessions not only mean lighter requirements for memory storage and protocol complexity, but also the availability of syncronised transmissions within the nodes in the cluster as encryption / authentication operations on transmitted data can be performed by the master at the same time using a unique key. \\
\end{enumerate}

Let us know give a formal description of the scheme. At the end of the following session, an example over a simple network is shown.
\subsection{Parameter definition}
Let $N$ be a positive integer and let $V \deq \{1,2, \ldots, N\}$. 
The \emph{authenticated network topology}  is a symmetric directed graph $\ANT = (V,E)$, i.e.~a graph where $E \subset V\times V$ and if $(i,j) \in E$, then $(j,i) \in E$. We furthermore assume that $\ANT$ is loop-free, i.e. without cycles of length 1. For each $1 \leq i \leq N$, $\ANTi{i} = (V_i, E_i)$ is the (non-symmetric and cycle-free) directed 
 subgraph of $\ANT$ such that 
\[E_i \deq \left\{ e\in E  \mid t(e) = i\right\} \,\,\text{ and } \,\,V_i \deq \{i\} \cup \{h(e)\mid e \in E_i\}.\] 
In the point of view of our application, $\ANTi{i}$ is the subgraph of the users which user $i$ is entitled to communicate to. An example of network topology network is depicted in Fig.~\ref{graph}.\\

\begin{figure}[h]
  \centering
  \begin{tikzpicture}[scale=0.8]
			\begin{pgfonlayer}{nodelayer}
				\node [style=view2] (0) at (-2.25, 3) {};
				\node [style=view2 ] (1) at (1.5, 3) {};
				\node [style=view2] (2) at (-0.75, 1.5) {\footnotesize $i$};
				\node [style=view] (3) at (-3.75, 1.5) {};
				\node [style=view] (4) at (3, 1.5) {};
				\node [style=view2] (5) at (0.5, -0.5) {};
				\node [style=view2] (6) at (-2, -0.5) {};
				\node [style=view2] (7) at (0, 4.25) {};
				\node [style=view] (8) at (-4.25, 4.25) {};
			\end{pgfonlayer}
			\begin{pgfonlayer}{edgelayer}
				\draw [style=arco] (0) to (3);
				\draw [style=arco] (0) to (2);
				\draw [style=arco] (1) to (2);
				\draw [style=arco] (1) to (2);
				\draw [style=arco] (1) to (4);
				\draw [style=arco] (2) to (5);
				\draw [style=arco] (2) to (6);
				\draw [style=arco] (3) to (6);
				\draw [style=arco] (7) to (2);
				\draw [style=arco] (8) to (0);
		\end{pgfonlayer}
	\end{tikzpicture}  
	\caption{An example of $\ANT$, where red nodes represent $\ANTi{i}$.}
  \label{graph}
\end{figure}
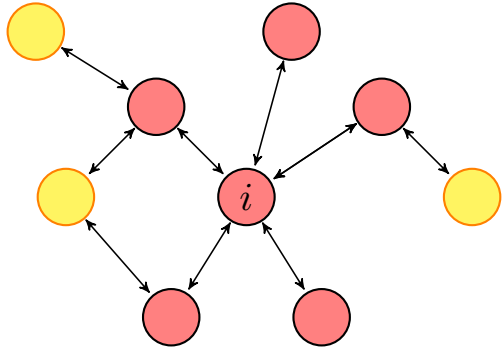
Let us now denote by $\EC$ an elliptic curve over $\F_q$, where $q > 3$, and $G \in \EC$ is the base point, whose order is prime number $p$. From now on we will assume $p >> N$.
Once the $\ANT$ has been established, the $\CA$ is in charge of the assignment of the scheme parameters to each node. For each node $i \in V$, its assigned \emph{local configuration data} $\LCD_{i} = (\Ss_i, \Pp_i)$ is such that
\[
\Ss_i \deq \{\bk_i, \bt_i\} \,\,\text{ and }\,\,\Pp_i \deq \bigl\{\bmm_{i-j}G\bigr\}_{j \in V_i \setminus \{i\}},
\]
where $\bk_i \in (\F_p)^2\setminus \{\bf 0\}$ is called the \emph{local key component} corresponding to the node $i$, $\bt_i \in (\F_p)^2\setminus \{\bf 0\}$ is called the \emph{transmitted key component} corresponding to the node $i$, and $\bmm_{i-j}G \in \EC^2\setminus \{0\}$ is called the \emph{topology vector} corresponding to the arrow $(i,j)$. The component $\Ss_i$ represents the private information assigned to node $i$, whereas $\Pp_i$ represents its public information.

\noindent The $\CA$ assigns the parameters to each node in a sequential way, once it has chosen an arbitrary \emph{root} node for each connected component of the graph. Starting from the parameters assigned to the root node, the $\CA$ computes the parameters for the other nodes of the graph according to some constraints which allow each pair of topologically admissible nodes to compute a shared secret, that we called \emph{Elliptic Curve Topology Authenticated Key} ($\ET$).  
\subsection{Parameter assignment and shared secrets}
Let $\ANT=(V,E)$ be the planned topology. As already mentioned, the parameter assignment is carried out in a sequential way by the $\CA$, starting from a root node in each connected component of the graph. 
\subsubsection{Parameters from a root node}\label{sec:ass}
Let node $i$ be the first root node chosen by the $\CA$.  Then the parameters $\bk_i$ and $\bt_i$ are generated randomly from $(\F_p)^2\setminus\{{\bf 0 }\}$ and assigned to the secret component $S_i = \{\bk_i,\bt_i\}$ of node $i$. Now, for each node $j$ to be connected to node $i$, two cases need to be distinguished:\\

\paragraph{\underline{$S_j$ \textbf{is not defined}}:}  the parameter $\bmm_{i-j}$ is generated randomly by the $\CA$, provided that $\bk_i \cdot \bmm_{i-j} \neq 0$, and the corresponding topology vector $\bmm_{i-j}G$ is appended to the public component $P_i$ of node $i$. Once the topology vector related to the arrow $(i,j)$ has been defined, the  parameters for node $j$ can be defined running the following steps:
\begin{itemize}
\item the parameter $\bk_j$ is randomly chosen by the $\CA$ from the solutions of the linear equation
\begin{equation}\label{eq:eqi-j}
\bk_i \cdot \bmm_{i-j} = \bk_j \cdot \bt_i
\end{equation}
and it is assigned to node $j$;
\item the parameter $\bmm_{j-i}$ related to the arrow $(j,i)$ is generated randomly, provided that \[\bk_j \cdot \bmm_{j-i} \neq 0;\]
the corresponding topology vector $\bmm_{j-i}G$ is assigned to node $j$;
\item  $\bt_j$ is randomly chosen by the $\CA$ from the solutions of the linear equation
\begin{equation}\label{eq:eqj-i}
\bk_j \cdot \bmm_{j-i} = \bk_i \cdot \bt_j
\end{equation}
and it is assigned to node $j$.
\end{itemize}

\noindent At the end of this process we have that
\begin{itemize}
\item  $\bmm_{i-j}G$ is appended to  $P_i$ for node $i$, 
\item $S_j = \{\bk_j, \bt_j\}$ and $\bmm_{j-i}G$ is appended to  $P_j$ for  node $j$.\\
\end{itemize}

\paragraph{\underline{$S_j$ \textbf{is already defined}}:} the topology vector related to the arrows $(i,j)$ and $(j,i)$ are defined as follows: 
\begin{itemize}
\item the parameter $\bmm_{i-j}$ is randomly chosen by the $\CA$ from the solutions of the linear equation
\begin{equation}\label{eq:eqi-j_arrow}
\bk_i \cdot \bmm_{i-j} = \bk_j \cdot \bt_i;
\end{equation}
\item the parameter $\bmm_{j-i}$ is randomly chosen by the $\CA$ from the solutions of the linear equation
\begin{equation}\label{eq:eqj-i_arrow}
\bk_j \cdot \bmm_{j-i} = \bk_i \cdot \bt_j.
\end{equation}
\end{itemize}
\noindent At the end of this process we have that
\begin{itemize}
\item $\bmm_{i-j}G$ is appended to  $P_i$ for node $i$,
\item $\bmm_{j-i}G$ is appended to  $P_j$ for  node $j$.
\end{itemize}
This process is completed when the $\CA$ has assigned secret and public components to each node of the graph.
\subsubsection{The shared secret}
\noindent Assume now that node $i$ wants to start a session with node $j$. Then node $i$ and node $j$ can agree on an ephemeral shared secret, performing the following operations:
\begin{itemize}
\item node $i$ generates a random non-zero element $\alpha \ra \F_p$;
\item node $i$ sends $\alpha \bt_i  G$ to node $j$;
\item node $i$ defines $\ET_{i-j} \deq \alpha \bk_i\cdot (\bmm_{i-j} G)$.
\end{itemize}
Now node $j$ can compute 
\[
\bk_j \cdot (\alpha \bt_i G) = (\bk_j \cdot \alpha \bt_i) G = (\alpha \bk_i\cdot \bmm_{i-j}) G = \alpha \bk_i\cdot (\bmm_{i-j} G)= \ET_{i-j},
\]
where the second equality is obtained from Eq.~\eqref{eq:eqi-j}. 
Consequently node $i$ and node $j$ have shared the non-zero secret $\ET_{i-j} \in \EC$.
Similarly, node $j$ can agree with node $i$ on the shared secret 
\[
\ET_{j-i} \deq \alpha \bk_j\cdot (\bmm_{j-i} G) = \bk_i \cdot (\alpha \bt_j G),
\] where $\alpha$ is again an ephemeral random chosen non-zero element in $\F_p$ generated by node $j$, and the second equality is derived from Eq.~\eqref{eq:eqj-i}.
\begin{remark}
For each $(i,j) \in E$, the component $\bmm_{i-j}$, which is generated by the $\CA$ in order to define the public component $\bmm_{i-j}G$, is not accessible by any user (belonging or not to  the network), unless they can solve the ECDL problem, as better explained in Section~\ref{sec:security}.
\end{remark}
\begin{remark}
When the session between node $i$ and node $j$ has timed out or is not anymore valid, node $i$ and node $j$ can again agree on a disposable shared secret by selecting a new random parameter $\alpha$. The same happens if node $i$ is damaged and needs to be replaced by another sensor. The $\CA$ assigns to the new node $i$ the same secret parameter of the former node $i$, and the communication with node $j$ is again established by selecting a new random parameter $\alpha$.
\end{remark}

\begin{remark}
The parameter assignment is highly scalar-product based. For this reason, it is important to point out that for each $i,j\in V$, the products of Eq.~\eqref{eq:eqi-j} defining $\ET_{i-j}$ are uniformly distributed over $\F_p\setminus\{0\}$, since, by definition, its inputs are non-zero elements of $(\F_p)^2$. Moreover, the secret components $\bk_i$ and $\bt_i$ can be chosen  from the $p$ solutions of Eq.~\eqref{eq:eqi-j} and Eq.~\eqref{eq:eqj-i} respectively.  In other words, the constraints defined in~Eq.~\eqref{eq:eqi-j} and Eq.~\eqref{eq:eqj-i} reduce the complexity of guessing the secret component $\bk_i$ (respectively $\bt_i$) from $p^2$ to $p$, since the value to be guessed needs to satisfy a linear equation. However, this does not represent a security issue since the parameter $p$ is chosen to be a secure parameters, therefore the security of the scheme should rely on the size of $p$ and not on the size of $p^2$.
\end{remark}
\begin{example}
As an example, let us describe the parameter assignment on the simple $\ANT$ of Fig.~\ref{graph3}, where 
\[V = \{1,2,3\} \text{ and } E = \{(1,2),(2,1),(1,3),(3,1),(2,3),(3,2)\},\] and the root node where parameter assignment starts is assumed to be node $1$. 

\begin{figure}[h]
  \centering
  \begin{tikzpicture}[scale=0.8]
			\begin{pgfonlayer}{nodelayer}
				\node [style=view2, dashed] (2) at (-0.75, 1.5) {\footnotesize $1$};
				\node [style=view2] (5) at (0.5, -0.5) {\footnotesize $3$};
				\node [style=view2] (6) at (-2, -0.5) {\footnotesize $2$};
			\end{pgfonlayer}
			\begin{pgfonlayer}{edgelayer}
				\draw [style=arco] (2) to (5);
				\draw [style=arco] (2) to (6);
				\draw [style=arco] (5) to (6);
		\end{pgfonlayer}
	\end{tikzpicture}  
	\caption{Key assignment on a simple ANT, the root node is highlighted.}
  \label{graph3}
\end{figure}
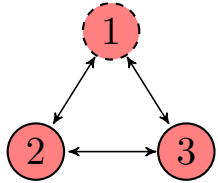
\subsubsection*{Node 1}
The parameters $\bk_1$ and $\bt_1$ are generated randomly from $(\F_p)^2\setminus\{{\bf 0 }\}$ and assigned to the secret component $S_1 = \{\bk_1,\bt_1\}$ of node $1$. 
\subsubsection*{Node 2}
Let us assume that the second node chosen by the $\CA$ is node $2$ in $V_1$. Then, the parameter $\bmm_{1-2}$ is generated randomly by the $\CA$, provided that $\bk_1 \cdot \bmm_{1-2} \neq 0$, and the corresponding topology vector $\bmm_{1-2}G$ is appended to the public component $P_1$ of node $1$. According to the description of Sec.~\ref{sec:ass}:
\begin{itemize}
\item the parameter $\bk_2$ is chosen by the $\CA$ from the solutions of 
\begin{equation*}\label{eq:eq1-2}
\bk_1 \cdot \bmm_{1-2} = \bk_2 \cdot \bt_1;
\end{equation*}
\item the parameter $\bmm_{2-1}$ related to the arrow $(2,1)$ is generated randomly, provided that \[\bk_2 \cdot \bmm_{2-1} \neq 0;\]
\item  $\bt_2$ is chosen by the $\CA$ from the solutions of \begin{equation*}\label{eq:eq2-1}
\bk_2 \cdot \bmm_{2-1} = \bk_1 \cdot \bt_2.
\end{equation*}
\end{itemize}
At the end of this process the $\CA$ has assigned:
\begin{itemize}
\item  $P_1 = \{\bmm_{1-2}G\}$ to node $1$,
\item $S_2 = \{\bk_2, \bt_2\}$ and $P_2 = \{\bmm_{2-1}G\}$ to node $2$.
\end{itemize}
\subsubsection*{Node 3}
Let us assume that the third node chosen by the $\CA$ is node $3$ in $V_1$. Then the secret component $S_1$ of  node $1$ remains unchanged and, proceeding as for node 2, the parameter $\bmm_{1-3}$ is generated randomly by the $\CA$, provided that $\bk_1 \cdot \bmm_{1-3} \neq 0$, and the corresponding topology vector $\bmm_{1-3}G$ is appended to the public component $P_1$ of node $1$. 
Once the topology vector related to the arrow $(1,3)$ has been defined, the parameters $\bk_3$, $\bmm_{3-1}$ and $\bt_3$  for node $3$ can be defined proceeding as for node $2$. At the end of this process the secret and public components are respectively
\begin{itemize}
\item  $S_1=\{\bk_1, \bt_1\}$ and  $P_1=\{\bmm_{1-2}G,\bmm_{1-3}G\}$ for  node $1$,
\item $S_2 = \{\bk_2, \bt_2\}$ and $P_2 = \{\bmm_{2-1}G\}$ for node $2$,
\item $S_3 = \{\bk_3, \bt_3\}$ and $P_3 = \{\bmm_{3-1}G\}$ for  node $3$.
\end{itemize}
\subsubsection*{Arrows $(2,3)$ and $(3,2)$}
In the topology under consideration, the communication between node $2$ and node $3$ is allowed. Since the secret components $S_2$ and $S_3$ have already been defined to link node 1 to node 2 and node 3, the parameter generation for $\bmm_{2-3}$ and $\bmm_{3-2}$ needs to be carried out as described in Sec.~\ref{sec:ass} in the second case. For this reason
\begin{itemize}
\item the parameter $\bmm_{2-3}$ is chosen by the $\CA$ from the solutions of 
\begin{equation*}
\bk_2 \cdot \bmm_{2-3} = \bk_3 \cdot \bt_2;
\end{equation*}
\item the parameter $\bmm_{3-2}$ is chosen by the $\CA$ from the solutions of 
\begin{equation*}
\bk_3 \cdot \bmm_{3-2} = \bk_2 \cdot \bt_3.
\end{equation*}
\end{itemize}
\noindent At the end of this process we have that
\begin{itemize}
\item $\bmm_{2-3}G$ is appended to  $P_2$ for node $2$,
\item $\bmm_{3-2}G$ is appended to  $P_3$ for  node $3$.
\end{itemize}
In conclusion, we have
\begin{itemize}
\item  $S_1=\{\bk_1, \bt_1\}$ and  $P_1=\{\bmm_{1-2}G,\bmm_{1-3}G\}$ for  node $1$,
\item $S_2 = \{\bk_2, \bt_2\}$ and $P_2 = \{\bmm_{2-1}G,\bmm_{2-3}G\}$ for node $2$,
\item $S_3 = \{\bk_3, \bt_3\}$ and $P_3 = \{\bmm_{3-1}G,\bmm_{3-2}G\}$ for  node $3$,
\end{itemize}
and the parameter generation is complete.
\end{example}
\subsection{ECTAK-based authenticated encryption}\label{sec:auth}
We show here a classical way to provide authenticated encryption, using the $\ET$ as shared secret. In the following we denote by $\Hh$ and $\KDF$ a keyed hash function and a key-derivation function respectively, whereas $\Enc_k$ and $\Dec_k$ respectively denotes the encryption and the decryption procedures using the key $k$ of a symmetric encryption method, where $\Dec_k = \Enc_k^{-1}$. \\
Assume now that node 1 wants to send a signed encrypted message $m$ to node 2. Then node 1 performs the following operations: 
\begin{itemize}
\item generates $\alpha \ra \F_p \setminus \{0\}$;
\item computes $\ET_{1-2} \deq \alpha \bk_1\cdot (\bmm_{1-2} G)$;
\item computes $(k_1,k_2) \deq \KDF(\ET_{1-2}, \alpha \bt_1 G)$;
\item computes $c \deq \Enc_{k_1}(m)$;
\item computes $s \deq \Hh(k_2,c)$;
\item sends $(\alpha \bt_1 G, c, s)$ to node 2,
\end{itemize}
where the size of $m,c$ and $s$ suits respectively the domain of encryption and hash functions.\\
\noindent Node 2 can perform the following steps:
\begin{itemize}
\item computes $ \bk_2 \cdot (\alpha \bt_1 G) = \ET_{1-2}$;
\item recovers $(k_1,k_2)$ by computing $\KDF(\ET_{1-2}, \alpha \bt_1 G)$;
\item recovers $m$ computing $\Dec_{k_1}(c)$;
\item checks that $s = \Hh(k_2,c)$.
\end{itemize}

\section{Considerations on security}\label{sec:security}
In this section we present some security properties of the scheme. We will show, in Theorem~\ref{th:main}, in which way an attacker can successfully determine the secret parameters of a target node $i$. Our interest in this case scenario is due to the occurrence that typical deployment of sensor nodes in the operation environment is generally unattended and therefore the risk of physical attacks such as the brute capture of a node should be seriously taken into consideration. We will prove that, in order the attack to be successful, the attacker needs to recover the secret information of at least two nodes connected to node $i$. The success probability of the attack is calculated in Theorem~\ref{th:lim}. The attack relies on the ability of the attacker to solve an instance of the intractable ECDL problem. To the best of our current state of knowledge, it is not possible to provide a formal reduction  from one problem to the other.\\

Let us now prove our result  showing that if an attacker can solve the ECDL problem  and can  recover the secret components $S_j$ and $S_k$, for some nodes $j,k \in V_i$ connected to $i$, then it can recover $S_i$ only if an algebraic condition on the coordinates is satisfied. In particular, from this follows that compromising one and only node is not enough to recover $S_i$.
Without any loss of generality, let us denote by node 1, node 2 and node 3 the three nodes previously mentioned. Let us assume that node 1 is targeted by the attacker, which has successfully recovered data from node 2 and node 3.
Moreover, to further simplify, let us assume  $\ANTi{1} = (V_1, E_1)$, where $V_1 = \{1,2,3\}$ and $E_1 = \{(1,2),(1,3)\}$, as depicted in Fig.~\ref{graph2}.
\begin{figure}[h]
  \centering
  \begin{tikzpicture}[scale=0.8]
			\begin{pgfonlayer}{nodelayer}
				\node [style=view2, dashed] (2) at (-0.75, 1.5) {\footnotesize $1$};
				\node [style=view2] (5) at (0.5, -0.5) {\footnotesize $3$};
				\node [style=view2] (6) at (-2, -0.5) {\footnotesize $2$};
			\end{pgfonlayer}
			\begin{pgfonlayer}{edgelayer}
				\draw [style=semiarco] (2) to (5);
				\draw [style=semiarco] (2) to (6);
		\end{pgfonlayer}
	\end{tikzpicture}  
	\caption{The network targeted by the attacker, where the target node 1 is highlighted.}
  \label{graph2}
\end{figure}
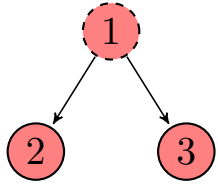

\noindent We assume that $\LCD_2 = \{(\bk_2,\bt_2),\bmm_{2-1}G\}$ and $\LCD_3 = \{(\bk_3,\bt_3),\bmm_{3-1}G\}$ are known to the attacker aiming at recovering $\bk_1,\bt_1 \in \LCD_1$, recalling that $\bmm_{1-2}G, \bmm_{1-3}G$ are publicly available.

\begin{eqnarray*}
(\bk_2 \cdot \bt_1)G &=& \bk_1 \cdot (\bmm_{1-2}G)  \\
(\bk_1 \cdot \bt_2)G &=& \bk_2 \cdot (\bmm_{2-1}G)\\
(\bk_3 \cdot \bt_1)G &=& \bk_1 \cdot (\bmm_{1-3}G)  \\
(\bk_1 \cdot \bt_3)G &=& \bk_3 \cdot (\bmm_{3-1}G).
\end{eqnarray*}
Since the attacker has access to an algorithm which solves the ECDL problem, it can access $\bmm_{1-2} = (a_{13},a_{14})$ and $\bmm_{1-3} = (a_{33},a_{34})$. Therefore,
denoting by $\bk_2 = (a_{11},a_{12})$, $\bt_2 = (a_{23},a_{24})$, $\bk_3 = (a_{31},a_{32})$ and $\bt_3 = (a_{43},a_{44})$, 
 the previous equations correspond to the following system  in the  unknowns $\bt_1 = (x_1,x_2)$ and $\bk_1 = (x_3,x_4)$:
\begin{equation}\label{eq:sys}
\left\{
\begin{array}{rl}
(a_{11}x_{1} + a_{12}x_2 - a_{13}x_3 -a_{14}x_4)G &= 0\\
(a_{23}x_3 + a_{24}x_4)G &= b_2G\\
(a_{31}x_{1} + a_{32}x_2 - a_{33}x_3 -a_{34}x_4)G &= 0\\
(a_{43}x_3 + a_{44}x_4)G &= b_4G,
\end{array}
\right.
\end{equation}
where $b_2 = \bk_2 \cdot \bmm_{2-1}$ and $b_4 = \bk_3 \cdot \bmm_{3-1}$. Therefore, denoting by 
\begin{equation}\label{eq:mat}
A \deq 
\begin{pmatrix}
a_{11} & a_{12} & -a_{13} & -a_{14}\\
 0 & 0 & a_{23} & a_{24}\\
a_{31} & a_{32} & -a_{33} & -a_{34}\\
 0 & 0 & a_{43} & a_{44}
\end{pmatrix},
\end{equation}
$b \deq (0,b_2,0,b_4)^t$ and $x \deq (x_1,x_2,x_3,x_4)^t$, the system in Eq.\eqref{eq:sys}~is equivalent to the linear equation 
\begin{equation}\label{eq:axb}
Ax=b.
\end{equation}

\begin{theorem}\label{th:main}
Let $\Aa$ be an adversary then can solve the ECDL problem. If $\det(A) \ne 0$, then $\Aa$ can recover $S_1$.
\end{theorem}
\begin{proof}
Since the adversary can solve the ECDL problem, it can build the system $Ax= b$ of Eq.\eqref{eq:axb}. The result trivially follows, since $\det(A) \ne 0$ implies that the system admits one and only one solution. 
\end{proof}
\begin{remark}
Notice that, due to the requirements on the $\LCD$ on each node, the first and the second row of $A$ are linearly independent, and the same holds for the third and the fourth.
If $\Aa$ can access the secret components of two nodes, then, since $\rk(A|b)=\rk(A)\in \{2,4\}$, the system in Eq.\eqref{eq:sys} as at least one solution. Moreover, if $\rk{A} = 2$, the system admits $p^2$ solutions, then the method of Theorem~\ref{th:main} leads to an attack to the scheme with success probability $1/{p^2}$. Indeed, if the attacker selects one of the $p^2-1$ solutions of the system which do not match the correct secret component of node 1, then the attempted impersonation attack is easily disclosed in the authenticated-encryption phase of the protocol (see Sec.~\ref{sec:auth}). The same holds if only one node is compromised by $\Aa$, since only two equations of the system are known. 
\end{remark}

We will now show that the success probability of the attack described in Theorem~\ref{th:main} approaches 1 when the prime $p$ is sufficiently large.
\begin{theorem}\label{th:lim}
Let $\mathcal S_p$ be the probability that $\Aa$ successfully recovers the secret component of $S_1$ using the method of Theorem~\ref{th:main}. Then
\[
\lim_{{p \to \infty}}\mathcal S_p = 1.
\]
\end{theorem}
\begin{proof}
Let us recall that, for each parameter assignment in a scheme with 3 nodes, we can construct a matrix as in Eq.~\eqref{eq:mat}. We call such a matrix an \emph{admissible} matrix for the scheme. Let us denote by $\mathcal A_p$ be the set of matrices in $(\F_p)^{4 \times 4}$ which are admissible, and by $\mathcal I_p$ the subset of those which are invertible. Then we have
\begin{equation}\label{eq:lim}
\mathcal S_p = \frac{\#\mathcal I_p}{\#\mathcal A_p} > \frac{\#\mathcal I_p}{p^{12}}.
\end{equation}
Let us now count $\#\mathcal I_p$.
The parameters $\bk_1$ and $\bt_1$ are chosen randomly in $(\F_p)^2\setminus\{{\bf 0}\}$, whereas $\bmm_{1-2} \in (\F_p)^2\setminus\{{\bf 0}\}$ is chosen such that $\bk_1 \cdot \bmm_{1-2} \ne 0$. The coefficient $\bk_2$ is chosen such that $\bk_2 \cdot \bt_1 \ne 0$, since $\bk_2 \cdot \bt_1 = \bk_1 \cdot \bmm_{1-2}$. Notice that, if $\bt_1 = (x_1,x_2)$, then
we need to rule out from the possible choices of $\bk_2=(a_{11},a_{12})$ those which satisfy $a_{11}x_{1} =- a_{12}x_2$. This reduces to $p^2-p$ the possibilities for the vector $\bk_2$.
Now, $\bmm_{2-1}$ can be chosen making sure that $\bk_2 \cdot \bmm_{2-1}\ne 0$. Since $\bk_2$ is not fixed, we have $p^2-1$ possible choices for $\bmm_{1-2}.$ Analogously, $\bt_2$ can be chosen in $p^2-p$ ways, since the value $\bk_1 = (x_3,x_4)$ is fixed and $\bk_1 \cdot \bt_2 \ne 0$. Hence, if the solution $(x_1,x_2,x_3,x_4) \in (\F_p)^4$ is fixed, we obtain $\bk_2, \bt_2, \bmm_{2-1}$ in $(p^2-p)^2(p^2-1)$ way. The same holds when considering $\bk_3, \bt_3, \bmm_{3-1}$. Noticing that in this argument we are using the fact that the constructed matrix is invertible, since we are assuming that $(x_1,x_2,x_3,x_4)$ is the unique solution of the problem, we obtain
\[
\mathcal I_p = (p^2-p)^4(p^2-1)^2 = \left(\#\GL(2,p)^2\right)(p^2-p)^2.
\] 
The result follows from Eq.~\eqref{eq:lim}, considering the limit for ${p \to \infty}$.
\end{proof}

Notice that $\ETS$ can be defined over an $\F_p$-vector space of dimension $d > 2$ similarly to the way it was built in Section~\ref{sec:scheme}  for an $\F_p$-vector space of dimension $2$, and Theorems~\ref{th:main} and~\ref{th:lim} can be extended to the $d$-dimensional case as well. In particular, it is possible to construct a $2d \times 2d$ triangular block matrix $A$ that, for a large $p$, is invertible with probability close to $1$. Moreover, an attacker who can solve the ECDL problem can successfully determine the secret parameters of a target node $i$, provided that it can recover the secret components of at least $d$ nodes connected to node $i$ and $\det(A)\ne 0$.

\section{Conclusion and future works}\label{sec:concl}
In this paper we have introduced the protocol $\ETS$, derived from \cite{pugliese08} and here adapted to the case of elliptic-curve cryptography. We have studied some security issues of the scheme, with a focus on the underlying ECDL problem. We have proven that, even though the secret and public components of the scheme are linked by means of linear equations, an attacker who wants to make use of the linear algebra method (explained in Sec.~\ref{sec:security}) to recover the secret components to a target node needs to be able to solve the ECDL problem and to access the secret components of at least two nodes connected to the target node. Although at the time of writing we understand that the scheme lacks of a general and complete security proof, the search for an argument showing that an attack to $\ETS$ can be converted into an attack to the underlying ECDL problem remains open.
\section*{Aknowledgment}
The authors wish to thank prof. Fortunato Santucci, director of the Centre of EXcellence on Connected Geolocalized and Cybersecure Vehicles (ExEMERGE) at University of L'Aquila, for his contributions in the development of TAKS and for stimulating the further research steps. Studies on the elliptic curve extension of TAKS and the on going implementation of ECTAKS algorithms on hardware chipset for fast automotive applications are part of the research activities in ExEMERGE.

\bibliographystyle{plain}
\bibliography{ECTAKS.bib}

\end{document}